\newcommand{\cmnts}{0} %
\newcommand{\nocommentsv}[1]{\ifthenelse{\equal{\cmnts}{0}}{#1}{}}
\newcommand{\commentsv}[1]{\ifthenelse{\equal{\cmnts}{0}}{}{#1}}
\newcommand{\cmt}[1]{\commentsv{\marginpar{{\color{red}{#1}}}}}

\newif\ifextended

\extendedtrue

\documentclass[a4paper,envcountsame]{llncs}

\usepackage{textcomp}
\usepackage{latexsym} 
\usepackage{amsfonts,amssymb,amsmath} 
\usepackage[english]{babel}
\usepackage{hyperref}
\usepackage{color}
\usepackage{general}
\usepackage{ptrs}

\title{On Probabilistic Term Rewriting\thanks{This work is partially supported by the ANR
    projects 14CE250005 ELICA, and 16CE250011 REPAS,
and the FWF project Y757.}}%

\author{Martin Avanzini\inst{1} \and Ugo {Dal Lago}\inst{1}\inst{2} and Akihisa Yamada\inst{3}}
\institute{
  INRIA Sophia Antipolis,
  France
  \and
  Department of Computer Science,
  University of Bologna, Italy
  \and
  National Institute of Informatics,
  Japan
}

\begin{document}

\maketitle

\begin{abstract}
We study the termination problem for probabilistic term rewrite
systems.
We prove that the interpretation method is
sound and complete for a strengthening of positive almost sure
termination, when abstract reduction systems and term rewrite
systems are considered.  Two instances of the interpretation method---%
polynomial and matrix interpretations---%
are analyzed and shown to capture interesting and nontrivial examples
when automated.
%
We capture probabilistic computation in a novel way by way of
multidistribution reduction sequences, this way accounting for both
the nondeterminism in the choice of the redex and the probabilism
intrinsic in firing each rule.
\end{abstract}

\section{Introduction}
Interactions between computer science and probability theory are
pervasive and extremely useful to the first discipline. Probability
theory indeed offers models that enable \emph{abstraction}, but it
also suggests a new model \emph{of computation}, like in randomized
computation or cryptography~\cite{GoldwasserMicali}. All this has
stimulated the study of probabilistic computational models and
programming languages: probabilistic variations on well-known models
like automata~\ifextended\cite{de1956computability,Rabin63}\fi,
Turing machines~\ifextended\cite{Santos69,gill77}\fi,
and the $\lambda$-calculus~\ifextended\cite{SahebDjahromi,JonesPlotkin}\fi
 are known from the early days
of theoretical computer science.

The simplest way probabilistic choice can be made available in
programming is endowing the language of programs with an
operator modeling sampling from (one or many) distributions. Fair,
binary, probabilistic choice is for example perfectly sufficient to
get universality if the underlying programming language is itself
universal (e.g., see \cite{dallagozorzi2012}).

\emph{Term rewriting}~\cite{Terese}
is a well-studied model of computation when no probabilistic behavior is involved.
It provides
a faithful model of pure functional programming which is, up to a
certain extent, also adequate for modeling higher-order parameter
passing~\cite{DalLagoMartini}. What is peculiar in term rewriting is
that, in principle, rule selection turns reduction into a potentially
nondeterministic process.  The following question is then a natural
one: is there a way to generalize term rewriting to a
fully-fledged \emph{probabilistic} model of computation?
Actually, not
much is known about probabilistic term rewriting:
the definitions we find
in the literature are
one by Agha et al.~\cite{AMS06} and
one by Bournez and Garnier~\cite{BG:RTA:05}.
We base our work on the latter,
where 
probabilistic rewriting is captured as a Markov decision process;
rule selection remains a nondeterministic process, but each rule
can have one of many possible outcomes, each with its own probability
to happen. Rewriting thus becomes a process in which both
nondeterministic and probabilistic aspects are present and
intermingled. When firing a rule, the reduction process implicitly
samples from a distribution, much in the same way as when performing
binary probabilistic choice in one of the models mentioned above.

In this paper, we first define a new, simple framework for discrete
probabilistic reduction systems, which properly generalizes standard
abstract reduction systems~\cite{Terese}.  In particular, what
plays the role of a reduction sequence, usually a (possibly
infinite) sequence $a_1\rightarrow a_2\rightarrow\ldots$ of \emph{states},
is a sequence
$\mu_1 \pTO \mu_2 \pTO \ldots$
 of \emph{(multi)distributions}
over the set of states. A multidistribution is not merely a
distribution, and this is crucial to appropriately account for both
the probabilistic behaviour of each rule and the nondeterminism in
rule selection. Such correspondence does not exist in Bournez and Garnier's
framework, as
nondeterminism has to be resolved by a \emph{strategy}, in order to define
reduction sequences. However, the two frameworks turn out to be
equiexpressive, at least as far as every rule has finitely many possible outcomes.
We then prove that the probabilistic ranking functions~\cite{BG:RTA:05} are 
sound and complete\footnote{%
The completeness of probabilistic ranking functions has been refuted in~\cite{FH:POPL:15},
but the counterexample there is invalid since a part of reduction steps
are not counted. We thank 
Luis Mar\'ia Ferrer Fioriti for this analysis.
}
 for proving \emph{strong almost sure termination},
a strengthening of \emph{positive almost sure termination}~\cite{BG:RTA:05}.
We moreover show that ranking functions provide bounds on expected runtimes.

This paper's main contribution, then, is the definition of a simple
framework for \emph{probabilistic term rewrite systems} as an example
of this abstract framework. Our main aim is studying whether any of
the well-known techniques for termination of term rewrite systems can
be generalized to the probabilistic setting, and whether they can be
automated. We give positive answers to these two questions, by
describing how polynomial and matrix interpretations can indeed be
turned into instances of probabilistic ranking functions, thus
generalizing them to the more general context of probabilistic term
rewriting.  We moreover implement these new techniques into the
termination tool \texttt{NaTT} \cite{NaTT}.

\section{Related Work}
Termination is a crucial property of programs, and has been widely
studied in term rewriting. Tools checking and certifying
termination of term rewrite systems are nowadays capable of
implementing tens of different techniques, and can prove termination
of a wide class of term rewrite systems, although the underlying
verification problem is well-known to be undecidable~\cite{Terese}. 

Termination remains an interesting and desirable property in a probabilistic
setting, e.g., in probabilistic programming~\cite{gmrbt2008} where
inference algorithms often rely on the underlying program to
terminate. But what does termination \emph{mean} when systems become
probabilistic? If one wants to stick to a \emph{qualitative}
definition, almost-sure termination is a well-known answer: a
probabilistic computation is said to almost surely terminate iff
non-termination occurs with null probability.  One could even
require \emph{positive} almost-sure termination, which
asks the expected time to termination to be finite.
Recursion-theoretically,
checking (positive) almost-sure termination is harder than checking
termination of non-probabilistic programs, where termination is at
least recursively enumerable, although undecidable: in a universal
probabilistic imperative programming language, almost sure termination
is $\Pi^0_2$ complete, while positive almost-sure termination is
$\Sigma^0_2$ complete \cite{KaminskiKatoen}.

Many sound verification methodologies for probabilistic termination
have recently been introduced (see,
e.g.,~\cite{BG:RTA:05,BG:RTA:06,Gnaedig07,FH:POPL:15,CFG:CAV:16}). In particular,
the use of ranking martingales has turned out to be quite successful
 when the analyzed program is imperative, and thus does not
have an intricate recursive structure. When the latter
holds, techniques akin to sized types have been shown to be applicable
~\cite{DalLagoGrellois}.
Finally, as already
mentioned, 
the current work can be seen as stemming from
the work by Bournez et al.~\cite{BournezKirchner,BG:RTA:05,BG:RTA:06}.
The added value compared to
their work are first of all the notion of multidistribution as way
to give an instantaneous description of the state of the underlying
system which exhibits both nondeterministic and
probabilistic features. Moreover, an interpretation method inspired by
ranking functions is made more general here, this way
acommodating not only interpretations over the real numbers, but also
interpretations over vectors, in the sense of matrix interpretations.
Finally, we provide an automation of polynomial and matrix interpretation inference
here, whereas nothing about implementation were presented in
Bournez's work.


\ifextended
\section{Preliminaries}
In this section, we give some mathematical preliminaries which will be
essential for the rest of the development.  With $\Real$ we denote the
set of \emph{real numbers}, with $\Realpos$ the set of
\emph{non-negative real numbers}, and with $\Realinfty$ the
set $\Realpos\cup\{\infty\}$. \cmt{def: multiset and notations
  missing (in particular disjoint sum, finiteness).}  \cmt{def:
  finitely branching relation.}  \cmt{def: proper = transitive and
  irreflexive.}

\paragraph{Probability Distributions.}
A \emph{(probability) distribution} on a countable set $A$ is a
function $\distone \ofdom A \to \Realpos$ such that $\sum_{a \in A}
\distone(a) = 1$.  The \emph{support} of a distribution $\distone$ is
the set $\Supp{\distone}\defsym\{\objone\in A\mid\distone(a)>0\}$.  We
write $\Dist{A}$ for the set of probability distributions over $A$.
We write $\prs{d(a_1):a_1,\dots;d(a_n):a_n}$ for distribution $d$ when
$\Supp{d}$ is a finite set $\{a_1,\dots,a_n\}$. 
pairwise distinct $a_i$s)

\ifextended
\paragraph{Stopping times.}
A \emph{stopping time} with respect to a stochastic process
$\RV=\{\rv_n\}_{n \in \N}$ is a random variable $S$, taking values in
$\N \cup \{\infty\}$, with the property that for each $t \in \N \cup
\{\infty\}$, the occurrence or non-occurrence of the event $S = t$
depends only on the values of $\seq[0][n]{\rv}$.  An instance of a
stopping time is the \emph{first hitting time} with respect to a set
$H$ is defined as $\tau_H(\omega) \defsym \min\{n \mid \rv_n(\omega)
\in H\}$, where $\min \varnothing = \infty$.  Every stopping time $S$
satisfies
\begin{equation}
  \label{eq:expected}
  \E(S) = \sum_{n=1}^\infty n \cdot \prob{S = n} = \sum_{n=1}^\infty \prob{S \geq n} \tpkt
\end{equation}
\fi


\fi
\section{Probabilistic Abstract Reduction Systems}
An \emph{abstract reduction system (ARS)}
on a set $\Cone$ is a binary relation
${\to} \subseteq \Cone \times \Cone$. 
Having $a \to b$ means
that $a$ reduces to $b$ in one step,
or $b$ is a one-step reduct of $a$. 
Bournez and Garnier~\cite{BG:RTA:05} extended the ARS formalism to probabilistic computations,
which we will present here using slightly different notations.

We write $\Realpos$ for the set of non-negative reals.
A \emph{(probability) distribution} on a countable set $A$ is a
function $d : A \to \Realpos$ such that $\sum_{a \in A} d(a) = 1$.
We say a distribution $d$ is \emph{finite} if
its \emph{support} $\Supp{d}\defsym\{a\in A\mid d(a)>0\}$
is finite,
and write $\prs{d(a_1):a_1,\dots;d(a_n):a_n}$ for $d$ if
$\Supp{d} = \{a_1,\dots,a_n\}$
(with pairwise distinct $a_i$s).
We write $\FDist{A}$ for the set of finite distributions on $A$.

\begin{definition}[PARS, \cite{BG:RTA:05}]
A \emph{probabilistic reduction} over a set $\Cone$ is a pair
of $a \in \Cone$ and $d \in \FDist{\Cone}$,
written $a \to d$.
A \emph{probabilistic ARS (PARS)} $\PARSone$ over $\Cone$ is
a (typically infinite) set of probabilistic reductions.
An object $\objone\in \Cone$ is called \emph{terminal} (or a \emph{normal form}) in $\PARSone$,
if there is no $\distone$ with 
$\objone \to \distone \in \PARSone$.
With $\Term{\PARSone}$ we denote the set of terminals in $\PARSone$.
\end{definition}

The intended meaning of $a \to d \in \PARSone$
is that
``there is a reduction step $a \to_\PARSone b$
with probability $\distone(b)$''.

\begin{example}[Random walk]\label{ex:rw}
A random walk over $\N$
with \emph{bias} probability $p$
is modeled
by the PARS $\Arw[p]$
 consisting of the probabilistic reduction
\[
  n + 1 \to \prs{p:n; 1-p:n+2} \quad \text{for all $n \in \N$.}
\]
\end{example}

A PARS describes both nondeterministic and probabilistic choice;
we say a PARS $\PARSone$ is \emph{nondeterministic}
if $a \to d_1, a \to d_2 \in \PARSone$ with $d_1 \neq d_2$.
In this case, the distribution of one-step reducts of $a$
is nondeterministically chosen from $d_1$ and $d_2$.
Bournez and Garnier~\cite{BG:RTA:05} describe reduction sequences via stochastic sequences,
which demand nondeterminism to be resolved by fixing a \emph{strategy} (also called \emph{policies}).
In contrast, we capture 
nondeterminism by defining a reduction
relation $\pTO[\PARSone]$ on distributions,
and emulate ARSs by $\prs{1:a} \pTO[\PARSone] \prs{1:b}$ when $a \to \prs{1:b} \in \PARSone$.
For the probabilistic case, taking Example~\ref{ex:rw}
we would like to have 
\[
  \prs{1:1} \pTO[{\Arw[\half]}] \prs{\half:0; \half:2} \tkom
\]
meaning that the distribution of one-step reducts of $1$ is
$\prs{\half:0; \half:2}$.
Continuing the reduction,
what should the distribution of two-step reducts of $1$ be?
Actually, it cannot be a distribution (on $A$):
by probability $\half$ we have no two-step reduct of $1$.
One solution, taken by \cite{BG:RTA:05}, is to introduce $\bot \notin A$
representing the case where no reduct exists.
We take another solution:
we consider
generalized distributions where probabilities may sum up to less than one,
allowing
\[
  \prs{1:1} \pTO[{\Arw[\half]}] \prs{\half:0; \half:2} \pTO[{\Arw[\half]}]
  \prs{\quoter:1;\quoter:3}\tpkt
\]
Further continuing the reduction, one would expect
$\prs{\eighth:0; \quoter:2; \eighth:4}$ as the next step,
but note that a half of the probability $\quoter$ of $2$
is the probability of reduction sequence $2 \to_{\Arw[\half]} 1 \to_{\Arw[\half]} 2$,
and the other half is of $2 \to_{\Arw[\half]} 3 \to_{\Arw[\half]} 2$.

\begin{example}
\label{ex:nd}
  Consider the PARS $\Amd$ consisting of the following rules:
  \begin{align*}
    \fun{a} & \to \prs{\half:\fun{b_1}; \half:\fun{b_2}} &
    \fun{b_1} & \to \prs{1:\fun{c}} & \fun{c} & \to \prs{1:\fun{d_1}}\\
    &&
    \fun{b_2} & \to \prs{1:\fun{c}} &
    \fun{c} & \to \prs{1:\fun{d_2}} \tpkt
  \end{align*}
 Reducing $\fun{a}$ twice always yields $\fun{c}$, so
 the distribution of the two-step reducts of $\fun{a}$
 is $\prs{1:\fun{c}}$.
 More precisely,
 there are two paths to reach $\fun{c}$:
 $\fun{a} \to_\Amd \fun{b_1} \to_\Amd \fun{c}$ and
 $\fun{a} \to_\Amd \fun{b_2} \to_\Amd \fun{c}$,
 each with probability $\half$.
 Each of them can be nondeterministically continued to $\fun{d_1}$
 and $\fun{d_2}$,
 so the distribution of three-step reducts of $\fun{a}$ is
 the nondeterministic choice among
 $\prs{1:\fun{d_1}}$,
 $\prs{\half:\fun{d_1}; \half:\fun{d_2}}$,
 $\prs{1:\fun{d_2}}$. 
 On the other hand, 
 if we defined the reduction relation 
 $\pTO[\Amd]$ in such a way that $\prs{1:\fun{c}}$
 reduced to $\prs{\half:\fun{d_1}; \half:\fun{d_2}}$, 
 then we would not be able to emulate ARSs
 by reducing $\prs{1:\fun{c}}$
 only to $\prs{1:\fun{d_1}}$ and $\prs{1:\fun{d_2}}$.
\end{example}

These analyses lead us to the following generalization of distributions.

\begin{definition}[Multidistributions]
A \emph{multidistribution} on $\Cone$ is a finite multiset $\mdistone$
of pairs of $a\in A$ and
$0 \le p \le 1$, written $p:a$,
such that \[
\sz{\mu} \defsym \sum_{\pr{p}{a} \in \mdistone} p \ \le\ 1 \tpkt
\]
We denote the set of multidistributions on $\Cone$ by $\FMDist{\Cone}$.
\end{definition}
Abusing notation, we identify
$\prs{p_1:a_1, \dots; p_n:a_n} \in \FDist{A}$ with multidistribution
$\prms{p_1:a_1, \dots; p_n:a_n}$ as no confusion can arise.
For a function $f: A \to B$, we often generalize the domain and range to multidistributions as follows:
\[ f\bigl(\prms{p_1:a_1,\dots;p_n:a_n}\bigr) \defsym \prms{p_1:f(a_1),\dots;p_n: f(a_n)}\tpkt
\]
The \emph{scalar multiplication}
of a multidistribution
is
$p \cdot \prms{q_1:a_1,\dots; q_n:a_n}
\defsym \prms{p\cdot q_1 : a_1, \dots; p\cdot q_n: a_n}$,
which is also a multidistribution if $0 \le p \le 1$.
More generally, multidistributions are closed under \emph{convex multiset unions},
defined as
$\biguplus_{i = 1}^n p_i\cdot\mu_i$ with
$p_1,\dots,p_n \ge 0$ and $p_1+\cdots+p_n \le 1$.

Now we introduce the reduction relation $\pTO[\PARSone]$ over multidistributions. 

\begin{definition}[Probabilistic Reduction]\label{d:reduction}
Given a PARS $\PARSone$,
we define the \emph{probabilistic reduction relation}
${\pTO[\PARSone]} \subseteq \FMDist{A} \times \FMDist{A}$ as follows:
\[
\infer{\mset{\pr{1}{\objone}} \pTO[\PARSone] \msetempty}{\objone\in\Term{\PARSone}}\qquad
\infer{\mset{\pr{1}{\objone}} \pTO[\PARSone] \distone}{\objone \to \distone \in \PARSone}\qquad
\infer{
	\biguplus_{i=1}^n p_i \cdot \mu_i \pTO[\PARSone]
	\biguplus_{i=1}^n p_i \cdot \rho_i
}{
	\mu_1 \pTO[\PARSone] \rho_1 & \dots &
	\mu_n \pTO[\PARSone] \rho_n
}
\]
In the last rule, we assume $p_1,\dots,p_n \ge 0$ and $p_1+\cdots+p_n \le 1$.
We denote by $\Seq\mu$ the set of all possible reduction sequences from $\mu$,
i.e.,
$\{\mu_i\}_{i\in\N} \in \Seq\mu$ iff $\mu_0 = \mu$ and $\mu_i \pTO[\PARSone] \mu_{i+1}$ for
any $i \in \N$.
\end{definition}
Thus $\mdistone \pTO[\PARSone] \mdisttwo$ if $\mdisttwo$ is obtained from 
$\mdistone$ by replacing every nonterminal $a$ in $\mdistone$ with 
all possible reducts with respect to some $a \to \distone \in \PARSone$,
suitably weighted by probabilities, and by 
removing terminals. The latter implies that 
$\sz{\mdistone}$ is not preserved during reduction: it decreases
by the probabilities of terminals.


To continue Example~\ref{ex:rw},
we have the following reduction sequence:
  \begin{align*}
    \prms{1:1}
    & \pTO[{\Arw[\half]}] \prms{\half:0;\half:2}
      \pTO[{\Arw[\half]}] \msetempty \uplus \prms{\quoter:1;\quoter:3} \\
    & \pTO[{\Arw[\half]}] \prms{\eighth:0;\eighth:2} \uplus \prms{\eighth:2;\eighth:4}
      \pTO[{\Arw[\half]}] \dots
  \end{align*}
The use of multidistributions resolves the issues indicated in 
Example~\ref{ex:nd} when dealing with nondeterministic systems.
  We have, besides others, the reduction 
  \[
    \prms{1:\fun{a}} 
    \pTO[\Amd] \prms{\half:\fun{b_1}, \half:\fun{b_2}}
    \pTO[\Amd] \prms{\half:\fun{c}, \half:\fun{c}}
    \pTO[\Amd] \prms{\half:\fun{d_1}, \half:\fun{d_2}}\tpkt
  \]
  The final step is possible 
  because $\prms{\half:\fun{c}, \half:\fun{c}}$ is not collapsed to $\prms{1:\fun{c}}$.

When every probabilistic reduction in $\PARSone$ is of form $a \to \prs{1:b}$ for some $b$,
then $\pTO[\PARSone]$ simulates the non-probabilistic ARS
via the relation $\prms{1:\cdot} \pTO[\PARSone] \prms{1:\cdot}$.
Only a little care is needed as normal forms are followed by $\varnothing$.

\begin{proposition}
Let ${\hookrightarrow}$ be an ARS and define $\PARSone$ by $a \to \prs{1:b} \in \PARSone$ iff
$a \hookrightarrow b$.
Then $\prms{1:a} \pTO[\PARSone] \mu$ iff
either $a \hookrightarrow b$ and $\mu = \prms{1:b}$ for some $b$,
or $\mu = \varnothing$ and $a$ is a normal form in $\hookrightarrow$.
\end{proposition}
\ifextended
\begin{proof}
For $\prms{1:a} \pTO[\PARSone] \mu$ only the first two rules of Definition~\ref{d:reduction}
are effective. Then the claim directly follows. 
\end{proof}
\fi

\subsection{Notions of Probabilistic Termination}

A binary relation $\to$ is called \emph{terminating} if it does not give rise to an infinite 
sequence $\objone_1 \to \objone_2 \to \dots$. In a probabilistic setting,
infinite sequences are problematic only if they occur with 
\emph{non-null} probability.
\begin{definition}[\AST]
A PARS $\PARSone$ is \emph{almost surely terminating} (\emph{\AST})
 if for any reduction sequence
$\{\mu_i\}_{i\in\N} \in \Seq\mu$,
it holds that
$
\lim_{n\to\infty} \sz{\mdistone_n} = 0
$.
\end{definition}
Intuitively,
$\sz{\mdistone_n}$ is the probability of having $n$-step reducts,
so its tendency towards zero indicates that infinite reductions occur with zero probability.

\begin{example}[Example~\ref{ex:rw} Revisited]
  The system $\Arw[p]$ is \AST for $p \leq \frac{1}{2}$, whereas 
  it is not for $p > \frac{1}{2}$.
  Note that although $\Arw[\half]$ is \AST, the expected number of reductions
  needed to reach a terminal is infinite. 
\end{example}

Let $\PARSone$ be a PARS
and $\vec\mu = \{\mu_i\}_{i\in\N} \in \Seq\mu$.
Following terminology from rewriting, we define the 
\emph{expected derivation length} $\adl(\vec\mu) \in \R \cup \{\infty\}$ of $\vec\mu$ by 
\[
  \adl(\vec\mu) \defsym \sum_{i\ge1} |\mu_i|\ .
\]
Intuitively, this definition is equivalent to taking mean length of 
terminal paths in $\vec\mu$.
The notion of \emph{positive almost sure termination (\PAST)}, 
introduced by Bournez and Garnier~\cite{BG:RTA:05}, 
constitutes a refinement of \AST\ 
demanding that the expected derivation length is finite
for every initial state $a$ and for \emph{every strategy}, i.e.,
for every reduction sequence $\vec\mu$ starting from $a$,
$\adl(\vec\mu)$ is bounded.
Without fixing a strategy, however, this condition does not ensure
bounds on the derivation length.

\begin{example}\label{ex:omega}
  Consider the (non-probabilistic) ARS on $\N \cup \{\omega\}$
  with reductions 
  $\omega \to n$ and $n+1 \to n$ for every $n \in \N$.
  It is easy to see that every reduction sequence is of finite length, 
  and thus, this ARS is \PAST.\@
  There is, however, no global bound on the length of reduction sequences starting from $\omega$.
\end{example}
Hence we introduce a stronger notion, which actually plays a more
essential role than \PAST.\@ It is based on a natural extension of
\emph{derivation height} from complexity analysis of term rewriting.

\begin{definition}[Strong \AST]
A PARS $\PARSone$ is \emph{strongly almost surely terminating} (\emph{\SPAST})
if the \emph{expected derivation height} $\adh[\PARSone](\objone)$
of every $\objone \in \Cone$ is finite,
where $\adh[\PARSone](a) \in \R \cup \{\infty\}$ is defined as
$  \sup_{\vec\mu \in \Seq{\prms{1:a}}} \adl(\vec\mu)$.
\end{definition}
In Example~\ref{ex:omega}, we make essential use of $\omega$ that
admits infinitely many one-step reducts.  Thus the ARS is not finitely
branching, and does not contradict the claims in \cite{BG:RTA:05}.
Nevertheless \PAST and \SPAST does not coincide on finitely branching PARSs.
The following example is found by an anonymous reviewer.
\begin{example}
Consider PARS $\mathcal{A}$ over
$\N \cup \set{a_n \mid n \in \N}$, consisting of
\begin{align*}
a_n &\to \prs{\half: a_{n+1}; \half: 0}&
a_n &\to \prs{1: 2^n\cdot n}&
n+1 &\to \prs{1: n}\tpkt
\end{align*}
Then $P$ is finitely branching and PAST, because every reduction sequence from $\prms{1:a_0}$
is one of the following forms:
\begin{itemize}
\item
$\vec\mu_0 = \prms{1:a_0} \pTO \prms{1:0}$
\item
$\vec\mu_{n} = \prms{1:a_0} \pTO^n \prms{\frac1{2^n}: a_{n}; \frac1{2^n}:0}
\pTO \prms{\frac1{2^n}: 2^{n}\cdot n} \pTO^{2^{n}\cdot n} \prms{\frac1{2^n}: 0}$
with $n = 1, 2, \dots$
\item
$\vec\mu_\infty = \prms{1:a_0} \pTO \prms{\frac1{2}: a_{1}; \frac1{2}:0} \pTO
\prms{\frac1{4}: a_{2};\frac1{4}: 0} \pTO \cdots$
\end{itemize}
and $\adl(\vec\mu_\alpha)$ is finite for each $\alpha \in \N \cup \set{\infty}$.
However, $\adh[\mathcal{A}](a_0)$ is not bounded,
since
$\adl(\vec\mu_{n}) = \frac1{2^0} + \cdots + \frac1{2^{n-1}} + \frac1{2^n} +
 \frac1{2^n}(2^{n}\cdot n) \ge n$.
\end{example}


\subsection{Probabilistic Ranking Functions} 
Bournez and Garnier~\cite{BG:RTA:05} generalized \emph{ranking
  functions}, a popular and classical method for proving termination
of non-probabilistic systems, to PARS. We give here a simpler but
equivalent definition of probabilistic ranking function, taking
advantage of the notion of multidistribution.

For a (multi)distribution $\mu$ over real numbers,
the \emph{expected value} of $\mu$ is denoted by
$\E(\mu) \defsym \sum_{\pr{p}{x}\in \mu} p \cdot x$.
A function $f : A \to \Real$ is naturally generalized to
$f : \FMDist{A} \to \FMDist{\Real}$,
so for $\mu \in \FMDist{A}$, $\E(f(\mu)) = \sum_{\pr{p}{x} \in \mu} p \cdot f(x)$.
For $\epsilon > 0$ we define the order $>_\epsilon$ on $\Real$ by 
$x >_\epsilon y$ iff $x \geq \epsilon + y$.

\begin{definition}\label{d:lyapunov}
  Given a PARS $\PARSone$ on $\Cone$,
  we say that a function $f : \Cone \to \Realpos$ is a
  \emph{(probabilistic) ranking function}
(sometimes referred to as \emph{Lyapunov} ranking function),
   if there exists $\epsilon > 0$ such that
   $\objone \to \distone \in \PARSone$ implies
   $f(\objone) >_\epsilon \E(f(\distone))$.
\end{definition}

The above definition slightly differs from the formulation in~\cite{BG:RTA:05}:
the latter demands
the \emph{drift} $\E(f(d)) - f(a)$ to at least $-\epsilon$,
which is equivalent to $f(a) >_\epsilon \E(f(d))$;
and allows any lower bound $\inf_{a \in A}f(a) > -\infty$,
which can be easily turned into $0$ by adding the lower bound to the ranking function.

We prove that a ranking function ensures
\SPAST and gives a bound on expected derivation length.
Essentially the same result can be found in~\cite{CFG:CAV:16},
but we use only elementary mathematics not requiring notions from probability theory.
We moreover show that this method is complete for proving \SPAST.

\begin{lemma}\label{l:lyapunov:step}
  Let $f$ be a ranking function for a PARS $\PARSone$.
  Then there exists $\epsilon > 0$ such that
$\E(f(\mdistone)) \geq \E(f(\mdisttwo)) + \epsilon \cdot \sz{\mdisttwo}$
whenever $\mdistone \pTO[\PARSone] \mdisttwo$.
\end{lemma}
\begin{proof}
  As $f$ is a ranking function for $\PARSone$,
  we have $\epsilon > 0$ such that
  $\objone \to \distone \in \PARSone$ implies
  $f(\objone) >_\epsilon \E(f(\distone))$.
  Consider $\mdistone \pTO[\PARSone] \mdisttwo$. 
  We prove the claim by induction on the derivation of 
  $\mdistone \pTO[\PARSone] \mdisttwo$.
  \begin{itemize}
  \item Suppose $\mdistone = \prms{1:a}$ and $a\in\Term{\PARSone}$.
    Then $\mdisttwo = \varnothing$ and 
    $\E(f(\mdistone)) \geq 0 = \E(f(\mdisttwo)) + \epsilon \cdot \sz{\mdisttwo}$ since
    $\E(f(\varnothing)) = \sz{\varnothing} = 0$. 
  \item Suppose $\mdistone = \prms{1:\objone}$ and $\objone \to \mdisttwo \in \PARSone$.
    From the assumption
    $\E(f(\mu)) = f(\objone) >_\epsilon \E(f(\mdisttwo))$, and
    as $\sz{\mdisttwo} = 1$ we conclude
    $\E(f(\mdistone)) \ge \E(f(\mdisttwo)) + \epsilon \cdot \sz{\mdisttwo}$.
  \item Suppose $\mdistone = \biguplus_{i=1}^n p_i \cdot \mdistone_i$,
    $\mdisttwo = \biguplus_{i = 1}^n p_i \cdot \mdisttwo_i$, and
    $\mdistone_i \pTO[\PARSone] \mdisttwo_i$ for all $1 \le i \le n$. 
    Induction hypothesis gives
    $\E(f(\mdistone_i)) \geq \E(f(\mdisttwo_i)) + \epsilon \cdot \sz{\mdisttwo_i}$. 
    Thus,
    \begin{multline*}
      \E(f(\mdistone))  = \sum_{i=1}^n p_i \cdot \E(f(\mdistone_i))
       \geq \sum_{i=1}^n p_i \cdot (\E(f(\mdisttwo_i)) + \epsilon \cdot \sz{\mdisttwo_i}) \\
       = \sum_{i=1}^n p_i \cdot \E(f(\mdisttwo_i)) +
         \epsilon \cdot \sum_{i=1}^n p_i \cdot \sz{\mdisttwo_i} 
       = \E(f(\mdisttwo)) + \epsilon \cdot \sz{\mdisttwo} \tpkt
    \tag*{\qed}
    \end{multline*}
  \end{itemize}
\end{proof}

\begin{lemma}\label{l:lyapunov:finseq}
  Let $f$ be a ranking function for PARS $\PARSone$.
  Then there is $\epsilon > 0$ such that 
  $\E(f(\mdistone_0)) \geq \epsilon \cdot \adl(\vec\mu)$
  for every $\vec\mu = \{\mu_i\}_{i\in\N} \in \Seq{\mdistone_0}$. 
\end{lemma}
\begin{proof}
  We first show
  $\E(f(\mdistone_m)) \geq \sum_{i=m+1}^n \sz{\mdistone_i}$ for every $n \geq m$, 
  by induction on $m - n$. 
  Let $\epsilon$ be given by Lemma~\ref{l:lyapunov:step}.
  The base case is trivial, so let us consider the inductive step. 
  By Lemma~\ref{l:lyapunov:step} and induction hypothesis we get
  \begin{align*}
    \E(f(\mdistone_m)) &\geq \E(f(\mdistone_{m+1})) + \epsilon \cdot \sz{\mdistone_{m+1}}
    \\
    &\geq \epsilon \cdot \sum_{i=m+2}^n \sz{\mdistone_i} + \epsilon \cdot \sz{\mdistone_{m+1}}
    = \epsilon \cdot \sum_{i=m+1}^n \sz{\mdistone_i} \tpkt
  \end{align*}
  By fixing $m=0$, we conclude that
  the sequence $\bigl\{\epsilon \cdot \sum_{i=1}^n\sz{\mu_i}\bigr\}_{n\ge1}$ is bounded by
  $\E(f(\mu_0))$,
  and so is its limit $\epsilon \cdot \sum_{i\ge1}\sz{\mu_i} = \epsilon \cdot \adl(\vec\mu)$.
  \qed
\end{proof}

\begin{theorem}\label{t:lyapunov:sound}
  Ranking functions are sound and complete for proving \SPAST.
\end{theorem}
\begin{proof}
  For soundness,
  let $f$ be a ranking function for a PARS $\PARSone$.
  For every derivation
  $\vec\mu$ starting from $\prms{1:a}$,
  we have
  $\adl(\vec\mu) \leq \frac{f(a)}{\epsilon}$
  by Lemma~\ref{l:lyapunov:finseq}.
  Hence,
  $\adh[\PARSone](a) \leq \frac{f(a)}{\epsilon}$, concluding that $\PARSone$ is \SPAST.
  
  For completeness, suppose that $\PARSone$ is \SPAST, and let $a \to d \in \PARSone$.
  Then we have $\adh[\PARSone](\objone) \in \Real$, and
\begin{align*}
\adh[\PARSone](a)\ =
\sup_{
 \vec\mu \in \Seq{\prms{1:a}}
}\adl(\vec\mu)\ 
&\ge \sup_{
 \vec\mu \in \Seq{d}
}
 (1 + \adl(\vec\mu))\\
&=\ 1 + \sup_{\vec\mu \in \Seq{d}} \adl(\vec\mu)\ \
 =\ 1 + \E(\adh[\PARSone](d))
  \tkom
\end{align*}
concluding $\adh[\PARSone](a) >_1 \E(\adh[\PARSone](d))$.
Thus, taking $\epsilon = 1$, $\adh[\PARSone]$ is a ranking function according to 
  Definition~\ref{d:lyapunov}.
  \qed
\end{proof}


\subsection{Relation to Formulation by Bournez and Garnier}
\renewcommand\tterm[1][\vec{X}]{T_{#1}}
As done by Bournez and Garnier~\cite{BG:RTA:05}, the dynamics of probabilistic systems are commonly defined as stochastic sequences, i.e., 
infinite sequences of random variables whose $n$-th variable represents the $n$-th reduct. 
A disadvantage of this approach is that
nondeterministic choices have to be \emph{a priori}
resolved by means of strategies,
making the rewriting-based understanding of nondeterminism inapplicable.
We now relate this formulation 
to ours, and see that the corresponding notions of \AST and \PAST
coincide.

We shortly recap central definitions from \cite{BG:RTA:05}.
We assume basic familiarity with stochastic processes, see e.g. \cite{Bremaud:99}. 
Here we fix a PARS $\PARSone$ on $\Cone$. 
A \emph{history} (of length $n+1$) is a finite sequence
$\vec{\objone} = \objone_0,\objone_1,\dots,\objone_n$ of objects from $\Cone$, and such a sequence is called
\emph{terminal} if $\objone_n$ is. 
A \emph{strategy} $\phi$ is a function from nonterminal histories to distributions such that
$\objone_n \to \phi(\objone_0,\objone_1,\dots,\objone_n) \in \PARSone$.  
A history $\objone_0,\objone_1,\dots,\objone_n$ is called
\emph{realizable under $\phi$} iff for every $0\leq i<n$, it holds that
$\phi(\objone_0,\objone_1,\dots,\objone_i)(\objone_{i+1})>0$. 

\begin{definition}[Stochastic Reduction,~\cite{BG:RTA:05}]
Let $\PARSone$ be a PARS on $A$ and $\bot\notin \Cone$ a special symbol.
A sequence of random variables
$\RV=\{\rv_n\}_{n\in\N}$ over $\Cone\cup\{\bot\}$
is a \emph{(stochastic) reduction} 
in $\PARSone$
(under strategy $\phi$)
if 
\begin{align*}
  \prob{\rv_{n+1}=\bot\mid\rv_n=\bot}&=1;\\
  \prob{\rv_{n+1}=\bot\mid\rv_n=\objone}&=1&&\text{if $\objone$ is terminal;}\\
  \prob{\rv_{n+1}=\bot\mid\rv_n=\objone}&=0&&\text{if $\objone$ is nonterminal;}\\
  \prob{\rv_{n+1}=\objone \mid\rv_n=\objone_n,\dots,\rv_0=\objone_0}&=\distone(\objone)
  &&\text{if $\phi(\objone_0,\dots,\objone_n)=\distone$,}
\end{align*}
where $\objone_0,\dots,\objone_n$ is a realizable nonterminal history under $\phi$.
\end{definition}
\ifextended
Notice that as an immediate consequence of the law of total probability we obtain
\begin{equation}
  \label{eq:pi_n}
  \prob{\rv_{n}=\objone_n} 
  = \sum_{\mathclap{\seq[0][n]{\objone} \in \Cone}} \prob{\rv_0=\objone_0,\dots,\rv_n=\objone_n}  \tpkt
\end{equation}
\fi
Thus, $\RV$ is set up so that trajectories correspond to 
reductions $\objone_0 \to_\PARSone \objone_1 \to_\PARSone \cdots$, 
and 
$\bot$ signals termination.
In correspondence, the derivation length is given by the \emph{first hitting time} to $\bot$:

\begin{definition}[AST/\PAST of~\cite{BG:RTA:05}]
  For $\RV =\{\rv_n\}_{n\in\N}$
  define the random variable $\tterm \defsym \min\{n \in \N \mid \rv_n = \bot\}$, where $\min\varnothing = \infty$ by convention.
  A PARS $\PARSone$ is \emph{stochastically AST} (resp.\ PAST) if
  for every stochastic reduction $\RV$ in $\PARSone$,
  $\prob{\tterm = \infty} = 0$
  (resp.\ $\E(\tterm)
  < \infty$).
\end{definition}

\ifextended
We will now see that stochastic AST/\PAST\ coincides with AST/\PAST.
To this end, we first clarify the correspondence of stochastic reductions and reductions over multidistributions. 
The quintessence of this correspondence is that any stochastic derivation $\RV = \{\rv_n\}_{n \in \N}$
translates to a reduction $\mdistone_0 \pTO[\PARSone] \mdistone_1 \pTO[\PARSone] \mdistone_2 \pTO[\PARSone] \dots $
so that the probabilities $p = \prob{\rv_0=\objone_0,\dots,\rv_n=\objone_n} > 0$ of realisable 
histories $\seq[0][n]{\objone}$ in $\RV$ are recorded in $\mdistone_n$, i.e.,  $(p : \objone_n) \in \mdistone_n$.
Equation~\eqref{eq:pi_n} then gives the correspondence between $\rv_n$ and $\mdistone_n$. 
With Lemmas~\ref{l:stor} and~\ref{l:rtos}, we make this correspondence precise.
Guided by \eqref{eq:pi_n}, we associate the multidistribution $\mdistone$ over $\Cone$ with the distribution over $\Cone \cup \{\bot\}$ 
such that $\clps{\mdistone}(\objone) \defsym \sum_{\pr{p}{\objone} \in \mdistone} p$ for all $\objone \in \Cone$, and 
$\clps{\mdistone}(\bot) \defsym 1 - \sz{\mdistone}$.

\paragraph*{Stochastic Reduction to Multidistribution Reduction.}
Here we fix a stochastic reduction $\RV$ in a PARS $\PARSone$. 
We show that $\RV$ corresponds to a multidistribution reduction,
assuming $X_0$ is finitely supported:
$\prob{\rv_0 = a} > 0$ for only finitely many $a$.

\newcommand{\hdist}[1]{\mdisttwo_{#1}}
For each $n \in \N$, we define the multidistribution $\hdist{n}$ over realizable histories of length $n$ by
\[
  \hdist{n} \defsym \mset{p : (\seq[0][n]{\objone}) \mid p = \prob{\rv_0=a_0,\dots,\rv_n=a_n} > 0} 
  \tpkt
\]

Note that $\hdist0$ is well defined since $\rv_0$ is finitely supported.
Then we can inductively show that $\hdist{n}$ is well defined, using the fact that
$\Supp{d}$ is finite for every $a \to d \in \PARSone$.
%

The following lemma clarifies then how the multidistributions $\hdist{n}$ evolve.
\begin{lemma}\label{l:ds}
  For each $n \in \N$, we have
  \[
    \hdist{n+1} =  
    \biguplus_{\pr{p}{(\seq[0][n]{\objone})} \in \hdist{n}, \objone_n \notin \Term{\PARSone}}
    \begin{array}{l}
      \mopen p \cdot \phi(a_0,\dots,a_n)(a_{n+1}) : (a_0,\dots,a_n,a_{n+1}) \\
      {}\mid \phi(a_0,\dots,a_n)(a_{n+1}) > 0 \mclose
    \end{array}
    \tpkt
  \]
\end{lemma}
\begin{proof}
  Fix a realizable history $(\seq[0][n]{\objone})$ and let $\distone = \phi(\seq[0][n]{\objone})$. 
  Notice that $(\seq[0][n]{\objone},\objone_{n+1})$ is realizable if $\distone(\objone_{n+1}) > 0$.
  Recall $\prob{ \rv_{n+1}=a_{n+1} \mid \rv_0=a_0,\dots,\rv_n=a_n} = \distone(a_{n+1})$. 
  By definition of conditional probability we thus have
  \[
    \prob{\rv_0=a_0,\dots,\rv_{n+1}=a_{n+1}}
    = \distone(\objone_{n+1})
    \cdot \prob{\rv_0=a_0,\dots,\rv_n=a_n} 
    \tkom
  \]
  as desired.
  \qed
\end{proof}

\begin{lemma}\label{l:stor}
  There exist multidistributions $\mdistone_0, \mu_1, \dots$ such that 
  for all $n \in \N$,
  $\clps{\mdistone_n}(\objone) = \prob{\rv_n=\objone}$ and $\mdistone_n \pTO[\PARSone] \mdistone_{n+1}$.
\end{lemma}
\begin{proof}
  From $\nu \in \FMDist{A^n}$ define multidistribution
 \[\last{\nu} \defsym \mset{ p:{a_n} \mid p:(a_0,\dots,a_n) \in \nu}
 \tpkt
 \]
 We show $\mu_n \defsym \last{\nu_n}$
 satisfies the desired properties. 
 It is easy to see that $\overline{\mu_n}(a) = \prob{\rv_n = a}$.
 We show $\last{\nu_n} \pTO \last{\nu_{n+1}}$.
 Consider arbitrary $\pr{p}{(a_0,\dots,a_n)} \in \nu_n$.
 If $a_n$ is terminal,
 we have
 \[
 \prms{1:a_n} \pTO[\PARSone] \varnothing
 \tkom
 \]
 and otherwise
 we have $a_n \to \phi(a_0,\dots,a_n) \in \PARSone$,
 so
 \[
 \prms{1:a_n} \pTO[\PARSone] \phi(a_0,\dots,a_n)
 \tpkt
 \]
 Combining them we get
 \begin{align*}
 \last{\nu_n} &=
 \biguplus_{p:(a_0,\dots,a_n)\in \nu_n} p \cdot \mset{1 : a_n} \\
 &\pTO[\PARSone]
 \biguplus_{p:(a_0,\dots,a_n)\in \nu_n, a_n \notin \Term{\PARSone}} p \cdot \phi(a_0,\dots,a_n) \\
 &= \last{\nu_{n+1}}
   \tpkt
 \end{align*}
 The last equation follows from Lemma~\ref{l:ds}.
\qed
\end{proof}
  
\paragraph*{Multidistribution Reduction to Stochastic Reduction.}
\newcommand{\hmdist}[1]{\mdistthree_{#1}}
For the inverse translation, let us fix a PARS $\PARSone$
and $M = \{\mu_i\}_{i\in\N}$ such that
$\mdistone_0 \pTO[\PARSone] \mdistone_1 \pTO[\PARSone] \dots$ and $\mdistone_0$ is a distribution.
We now map $M$ to a stochastic sequence by fixing a strategy $\phi_M$ according to $M$. 
As a first step, let us construct a sequence of multidistributions $\hmdist{n}$ over realizable histories from $M$
in such a way that the multidistribution $\hmdist{n}$ assigns the 
probability $p$ to the history $\seq[0][n]{\objone}$ precisely when 
the occurrence $p: \objone_n \in \mdistone_n$ was developed through a sequence of steps 
$\objone_0 \to_\PARSone \dots \to_\PARSone \objone_n$ in $\mdistone_0 \pTO[\PARSone] \cdots \pTO[\PARSone] \mdistone_n$. 
\begin{definition}
  For each $n \in \N$ such that there is a step $\mdistone_{n-1} \pTO[\PARSone] \mdistone_{n}$ in $M$, 
  we define the multidistribution $\hmdist{n}$ over realizable histories of length $n$, with $\last{\hmdist{n}} = \mdistone_n$, inductively as follows. 
  Here, $\last{\hmdist{n}}$ is defined as in Lemma~\ref{l:stor}.
  First, we set $\hmdist{1} \defsym \mdistone_0$.
  In the inductive case, observe that for each nonterminal history $(a_0,\dots,a_n)$ occurring in $\hmdist{n}$ 
  there exists a transition $a_n \to d_{a_0,\dots,a_n} \in \PARSone$
  so that 
  \begin{align}
    \mdistone_n &= \biguplus_{p:(a_0,\dots,a_n) \in \hmdist{n}} \prms{p:a_n}
 \tag*{}\\
    \quad&\pTO[\PARSone]
    \biguplus_{\substack{p:(a_0,\dots,a_n) \in \hmdist{n}\\a_n\notin\Term{\PARSone}\\a_{n+1} \in \Supp{d_{a_0,\dots,a_n}}}}\prms{ p \cdot d_{a_0,\dots,a_n}(a_{n+1}): a_{n+1} } 
    = \mdistone_{n+1}
    \tpkt
\label{eq:mu-step}
  \end{align}
  We set 
  \[
    \hmdist{n+1} \defsym 
  \biguplus_{\substack{p:(a_0,\dots,a_n) \in \hmdist{n}\\a_n\notin\Term{\PARSone}\\a_{n+1} \in \Supp{d_{a_0,\dots,a_n}}}}\mset{ p \cdot d_{a_0,\dots,a_n}(a_{n+1}) : (a_0,\dots,a_{n+1}) }
  \tpkt
  \]
\end{definition}
Note that when $\mdistone_n = \varnothing$ we have $\hmdist{n+1} = \varnothing$. 

Crucially, even if two objects occurring in $\mdistone_n$ are equal, they are separated by their history in $\hmdist{n}$. 
In other words:
\begin{lemma}
  For all $n \in \N$ with $\hmdist{n}$ defined, we have that $\Supp{\hmdist{n}}$ is a set. 
\end{lemma}
\begin{proof}
  The proof is by induction on $n$. 
  The base case is trivial, as $\mdistone_0$ and thus $\hmdist{0}$ is a distribution. 
  The inductive step follows directly from the induction hypothesis and the fact that the supports $\Supp{\distone}$, 
  for the distributions $\distone$ mentioned in~\eqref{eq:mu-step}, are sets.
  \qed
\end{proof}
This then justifies the following definition of the strategy $\phi_M$, which we use in the simulation of $M$ below.
\begin{definition}
  We define the strategy $\phi_M$ so that
  $\phi_M(a_0,\dots,a_n) = \distone_{a_0,\dots,a_n}$,
  such that,
  if $(a_0,\dots,a_n)$ is a nonterminal histories,
  then $a_n \to \distone_{a_0,\dots,a_n} \in \PARSone$ is 
  used in~\eqref{eq:mu-step} to reduce the occurrence of $a_n$ in $\mdistone_n$,
  and otherwise $\distone$ is arbitrary.
\end{definition}

\begin{lemma}[Reduction to Stochastic Sequences]\label{l:rtos}
  Let $\RV=\{\rv_n\}_{n\in\N}$ be the stochastic derivation under strategy $\phi_M$
  with $\prob{\rv_0 = a} = \mu_0(a)$.
  Then
  $\prob{\rv_n = \objone} = \clps{\mdistone_n}(\objone)$ for all $n \in \N$.
\end{lemma}
\begin{proof}
  We show that $p:(\seq[0][n]{\objone}) \in \hmdist{n}$ iff $\prob{\rv_0=\objone_0,\dots,\rv_n=\objone_n} = p > 0$. 
  Using that $\last{\hmdist{n}} = \mdistone_n$, the lemma follows then from~\eqref{eq:pi_n}. 
  The proof is by induction on $n$. 

  The base case is trivial, as $\hmdist{0} = \mdistone_0$ corresponds to the starting distribution $\distone$ of $\RV$. 
  Concerning the inductive step, it suffices to realise that
  $p:(\seq[0][n]{\objone}) \in \hmdist{n}$ if and only if 
  $p \cdot q: (\seq[0][n+1]{\objone}) \in \hmdist{n+1}$, $\objone_n$ is nonterminal 
  and $q \defsym \phi_M(\seq[0][n]{\objone})(\objone_{n+1}) > 0$. 
  As $\prob{\rv_0=\objone_0,\dots,\rv_n=\objone_n,\rv_{n+1}=\objone_{n+1}} = q \cdot \prob{\rv_0=\objone_0,\dots,\rv_n=\objone_n}$ by definition, 
  the lemma follows then from induction hypothesis.
  \qed
\end{proof}

\paragraph*{Relating AST and \PAST\ to its stochastic versions.}
We have established a one-to-one correspondence between infinite multidistribution reductions $M$ and stochastic reductions $\RV$.
It is then not difficult to establish a correspondence between the expected derivation length of $M$ 
and the time of termination of $\RV$, relying on the following auxiliary lemma. 

\begin{lemma}\label{l:edl:aux}
  Let $\RV = \{\rv_n\}_{n \in \N}$ be a
  stochastic derivation in $\PARSone$
  with finitely supported $\rv_0$,
  and $M = \{ \mdistone_n \}_{n \in \N}$ a sequence of multidistributions 
  satisfying $\prob{\rv_n = \objone} = \clps{\mdistone_n}(\objone)$. 
  The following two properties hold.
  \begin{enumerate}
  \item\label{l:edl:aux:fin} 
    $\prob{\tterm \geq n} = \sz{\mdistone_n}$ for every $n \in \N$.
  \item\label{l:edl:aux:inf}
    $\prob{\tterm = \infty} = \lim_{n \to \infty} \sz{\mdistone_n}$.
  \end{enumerate}
\end{lemma}
\begin{proof}
  Concerning the first property, we have 
  \[
    \prob{\tterm \geq n} 
    = \prob{\rv_n \in \Cone}
    = \sum_{\objone \in \Cone} \prob{\rv_n = \objone}
    = \sum_{\objone \in \Cone} \clps{\mdistone_n}(\objone)
    = \sz{\mdistone_n}
    \tkom
  \]
  for all $n \in N$, where the penultimate equation follows from the assumption, and the last from the definition of $\clps{\mdistone_n}$. 
  As we have 
  \[
    \prob{\tterm = \infty} 
    = \lim_{n \to \infty} \prob{\tterm \geq n} 
    \tkom
  \]
  the second property thus follows from the first.
  \qed
\end{proof}
\fi

\begin{theorem}
  A PARS $\PARSone$ is (P)AST if and only if it is stochastically (P)AST.\@
\end{theorem}
\ifextended
\begin{proof}
  We consider the ``if'' direction first. 
  Suppose $\PARSone$ is AST.\@ 
  Lemma~\ref{l:stor} translates an 
  arbitrary stochastic derivation $\RV=\{\rv_n\}_{n \in \N}$ in $\PARSone$
  to a reduction
  $M = (\mdistone_0 \pTO[\PARSone] \mdistone_1 \pTO[\PARSone] \cdots)$, 
  for which we have $\prob{\tterm = \infty} = \lim_{n \to \infty} \sz{\mdistone_n} = 0$
  by Lemma~\eref{l:edl:aux}{inf}. Hence, $\PARSone$ is stochastically AST. 
  If $\PARSone$ is moreover \PAST, using Lemma~\eref{l:edl:aux}{fin} 
  and~\eqref{eq:expected} we 
  get 
  \[
    \N \ni \adl(M) 
    = \sum_{n \geq 1} \sz{\mdistone_n}
    = \sum_{n \geq 1} \prob{\tterm \geq n}
    = \E(\tterm)
    \tkom
  \]
  where we tacitly employ $\prob{\tterm = \infty} = 0$.   

  The ``only if'' direction is proven dual, using Lemma~\ref{l:rtos}.
  \qed
\end{proof}
\fi


\section{Probabilistic Term Rewrite Systems}
Now we formulate probabilistic term rewriting following \cite{BG:RTA:05}, and then lift the
interpretation method for term rewriting to the probabilistic case.

We briefly recap notions from rewriting; see \cite{BN:1998} for an
introduction to rewriting.  A \emph{signature} $F$ is a set of
\emph{function symbols} $\fs$ associated with their \emph{arity}
$\arity{\fs} \in \N$.  The set $\TERMS$ of \emph{terms} over a signature
$\FS$ and a set $\VS$ of variables (disjoint with $\FS$) is the least set
such that $\vx \in \TERMS$ if $\vx \in \VS$ and
$\fs(t_1,\dots,t_{\arity{\fs}}) \in \TERMS$ whenever $\fs \in \FS$ and $t_i
\in \TERMS$ for all $1 \leq i \leq \arity{\fs}$.
A \emph{substitution} is a mapping $\sigma : \VS \to \TERMS$,
which is extended homomorphically to terms.
We write $t\sigma$ instead of $\sigma(t)$.
A \emph{context} is a term $C \in \TERMS[\FS][\VS\cup\{\hole\}]$ 
containing exactly one occurrence of a special variable $\hole$. 
With $C[t]$ we denote the term obtained by replacing $\hole$ in $C$ with $t$.

To define probabilistic rewriting,
we first extend substitutions and contexts to multidistributions as before: 
$\mu\sigma \defsym \prms{p_1: t_1\sigma,\dots; p_n:t_n\sigma}$
and $C[\mu] \defsym \prms{p_1: C[t_1], \dots; p_n: C[t_n]}$
for $\mu = \prms{p_1:t_1, \dots; p_n:t_n}$.
Given a multidistribution $\mu$ over $A$,
we define a mapping $\overline{\mu} : A \to \Realpos$ by
$\overline{\mu}(a) \defsym \sum_{p:a\in \mu}p$,
which forms a distribution if $|\mu| = 1$.

\begin{definition}[Probabilistic Term Rewriting]
  A \emph{probabilistic rewrite rule} is a pair of $l \in \TERMS$ and
  $\distone \in \FDist{\TERMS}$, written $l \to \distone$.  A
  \emph{probabilistic term rewrite system (PTRS)} is a (typically
  finite) set of probabilistic rewrite rules.  We write
  $\widehat\PTRSone$ for the PARS consisting of a probabilistic reduction
  $C[l\sigma] \to C[\overline{d\sigma}]$ for every probabilistic rewrite rule
  $l\to d \in \PTRSone$, context $C$, and substitution $\sigma$.
    We
  say a PTRS $\PTRSone$ is AST/SAST if $\widehat\PTRSone$ is.
\end{definition}

Note that,
for a distribution $d$ over terms, $d\sigma$ is in general a multidistribution;
e.g., consider $\prs{\half:x; \half:y}\sigma$ with $x\sigma = y\sigma$.
This is why we need $C[\overline{d\sigma}]$, which is a distribution,
to obtain a probabilistic reduction above.

\begin{example}\label{e:trs}
\def\f{\fun{f}}
\def\s{\fun{s}}
\def\z{\fun{0}}
The random walk of Example~\ref{ex:rw} can be modeled by a PTRS
consisting of a single rule $\s(x) \to \prs{p:x; 1-p:\s(\s(x))}$.  To
rewrite a term, typically there are multiple choices of a subterm to
reduce (\emph{redex}), which is up to strategy.  For instance,
$\s(\f(\s(\z)))$ has two possible reducts: $\prms{p:\f(\s(\z)); 1-p:
  \s(\s(\f(\s(\z))))}$ and $\prms{p:\s(\f(\z)); 1-p:
  \s(\f(\s(\s(\z))))}$.
\end{example}


\subsection{Interpretation Methods for Proving \SPAST}
\newcommand{\Int}[2][\alpha]{\left\llbracket#2\right\rrbracket_\PAone^{#1}}
\newcommand{\IntD}{\interpret[\PAone]}
\newcommand{\I}[2][\PAone]{{#2}_{#1}}

We now generalise the \emph{interpretation method} for term rewrite systems to the probabilistic setting. The following notion is standard. 

\newcommand{\GTD}{\sqsupset}

\begin{definition}[$\FS$-Algebra]
  An \emph{$\FS$-algebra} $\PAone$ on a non-empty carrier set $X$
  specifies the \emph{interpretation} $\I\fs : X^{\arity\fs} \to X$ of
  each function symbol $\fs\in\FS$.
  We say $\PAone$ is \emph{monotone} with respect to a binary relation
  ${\GT} \subseteq X \times X$
  if $x \GT y$ implies
  $\I\fs(\dots,x,\dots) \GT \I\fs(\dots,y,\dots)$
  for every $\fs \in \FS$.
  Given an \emph{assignment} $\alpha \ofdom \VS \to X$,
  the interpretation 
  of a term is
  defined as follows:
  \[
    \Int{t} 
    \defsym
      \begin{cases}
        \alpha(t) & \text{if $t \in \VS$,}\\
        \I{\fs}(\Int{t_1},\dots,\Int{t_n})
        & \text{if $t = \fs(\seq[n]{t})$.}
      \end{cases}
    \]
  We write $s \GT_\PAone t$ iff $\Int{s} \GT \Int{t}$ for every assignment $\alpha$.
\end{definition}

In the non-probabilistic case, the interpretation method refers to a class of termination 
methods that use monotone $\FS$-algebras to embed reduction sequences into a well-founded order $\succ$.
This method is also complete:

\begin{theorem}[cf.\ \cite{Terese}]
A TRS $\TRSone$ is terminating iff
there exists an $\FS$-algebra $\PAone$
which is monotone with respect to a well-founded order $\GT$
and satisfies $\TRSone \subseteq {\GT_\PAone}$.
\end{theorem}
For a completeness proof, the \emph{term algebra} $\TA$, an
$F$-algebra on $\TERMS$ such that $\I[\TA]\fs(\seq[n]{t}) =
\fs(\seq[n]{t})$, plays a crucial role.  In this term algebra,
assignments are substitutions, and
$\interpret[\TA][\sigma]{t} = t\sigma$.

Now we introduce a probabilistic version of interpretation methods,
which is sound and complete for proving \SPAST.
To achieve completeness, we first keep the technique as general as possible.
  For an $\FS$-algebra $\PAone$,
  we lift the interpretation of terms to multidistributions as before, i.e., 
  \[
    \Int{\prms{p_1:t_1,\dots; p_n:t_n}} \defsym
    \prms{p_1:\Int{t_1},\dots; p:\Int{t_n}} \tpkt
  \]

\begin{definition}[Probabilistic $\FS$-Algebra]
A \emph{probabilistic monotone $\FS$-\linebreak algebra}
$(\PAone,{\GTD})$
is an $\FS$-algebra $\PAone$ equipped with
a relation ${\GTD} \subseteq X \times \FDist{X}$, such that 
for every $\fs \in \FS$, $\I{\fs}$ is monotone
with respect to $\GTD$, i.e., $x \GTD d$ implies
$\I{\fs}(\dots,x,\dots) \GTD \overline{\I{\fs}(\dots,d,\dots)}$
where $\I{\fs}(\dots,\cdot,\dots)$
is extended to (multi-)\linebreak distributions. 
\end{definition}

\ifextended
The following is an extension of a notion from \cite{HM:RTA:14}. 

\begin{definition}[Collapsibility]
  We say a relation ${\GTD} \subseteq X \times \FDist{X}$
  is \emph{collapsible} if 
  there exist a function $\Gcol : X \to \Realpos$ and $\epsilon > 0$ such that
  $x \GTD d$ implies
  $\Gcol(x) >_\epsilon \E(\Gcol(d))$.
  We say a monotone $\FS$-algebra $(\PAone,\GTD)$ is collapsible if $\GTD$ is. 
\end{definition}
\fi

For a relation ${\GTD} \subseteq X \times \FDist X$,
we define the relation ${\GTD_\PAone} \subseteq \TERMS \times \FDist\TERMS$
by
$t \GTD_\PAone d$
iff $\Int{t} \GTD \overline{\Int{d}}$ for every assignment $\alpha : \VS \to X$.
The following property is easily proven by induction. 
\begin{lemma}\label{l:inter}
  Let $(\PAone,\GTD)$ be a probabilistic monotone $\FS$-algebra.
  If $s \GTD_\PAone d$ then
$\Int{s\sigma} \GTD \overline{\Int{d\sigma}}$
and
$\Int{C[s]} \GTD \overline{\Int{C[d]}}$
for arbitrary $\alpha$,
$\sigma$, and
$C$.
\end{lemma}
\ifextended
\begin{proof}
Let $d = \prs{p_1:t_1,\dots; p_n:t_n}$.
Concerning the first property, define\ the assignment $\beta$ by
$\beta(x) = \Int{x\sigma}$ for every $x \in \VS$. 
By structural induction on $t$, one can verify $\Int[\beta]{t} = \Int{t\sigma}$ for all terms $t$.
Thus, from the assumption we get
\begin{align*}
  \Int{s\sigma}  = \Int[\beta]{s}
   \GTD \overline{\Int[\beta]{d}}
   &= \overline{\prms{p_1:\Int[\beta]{t_1},\dots; p_n: \Int[\beta]{t_n}}}\\
   &= \overline{\prms{p_1:\Int{t_1\sigma},\dots;p_n:\Int{t_n\sigma}}}
  = \overline{\Int{d\sigma}}
  \tpkt
\end{align*}
The second property is proven by induction on $C$, where 
the base case follows directly from the assumption, and the inductive step
from monotonicity.
\qed
\end{proof}
\fi

\begin{theorem}[Soundness and Completeness]\label{t:interpretation}
  A PTRS $\PTRSone$ is \SPAST\ iff
  there exists a collapsible monotone $\FS$-algebra $(\PAone,\GTD)$
  such that $\PTRSone \subseteq {\GTD_\PAone}$.
\end{theorem}
\begin{proof}
  For the ``if'' direction, we show that the PARS $\widehat\PTRSone$ is \SPAST
  using Theorem~\ref{t:lyapunov:sound}.
  Let $\alpha : \VS \to X$ be an arbitrary assignment,
  which exists as $X$ is non-empty. 
  Consider $s \to d \in \widehat\PTRSone$.
  Then we have
  $s = C[l\sigma]$ and $d = C[\overline{d'\sigma}]$ 
  for some $\sigma$, $C$, and $l \to d' \in \PTRSone$.
  By assumption we have $l \GTD_\PAone d'$, and
  thus $\Int{s} \GTD \overline{\Int{d}}$ by Lemma~\ref{l:inter}.
  The collapsibility of $\GTD$ gives a function
  $\Gcol : X \to \Realpos$ and $\epsilon > 0$
  such that
  $\Gcol(\Int{s}) >_\epsilon \E(\Gcol(\overline{\Int{d}}))$,
  and by extending definitions
  we easily see $\E(\Gcol(\overline{\Int{d}})) = \E(\Gcol(\Int{d}))$.
  Thus $\Gcol(\Int{\cdot})$ is a ranking function.

  For the ``only if'' direction, suppose that $\PTRSone$ is \SPAST.\@ 
  We show $(\TA,\widehat\PTRSone)$ forms a collapsible probabilistic monotone
  $\FS$-algebra orienting $\PTRSone$. 
  \begin{itemize}
  \item
    Since $\PTRSone$ is \SPAST, Theorem~\ref{t:lyapunov:sound} gives
    a ranking function $f \ofdom \TERMS \to \Realpos$ and $\epsilon > 0$
    for the underlying PARS $\widehat\PTRSone$.
    Taking $\Gcol = f$, $\widehat\PTRSone$ is
    collapsible.
  \item
    Suppose $s \mathrel{\widehat\PTRSone} d$.
    Then we have $s = C[l\sigma]$ and $d = C[\,\overline{d'\sigma}\,]$ for some $C$, $\sigma$,
    and $l \to d' \in \PTRSone$.
    As
    $f(\dots,C,\dots)$ is also a context,
    $f(\dots,s,\dots) \mathrel{\widehat\PTRSone} f(\dots,d,\dots)$,
    concluding monotonicity.
  \item
    For every probabilistic rewrite rule $l \to \distone \in \PTRSone$ and
    every assignment (i.e., substitution) $\sigma : \VS \to \TERMS$,
    we have 
    $\interpret[\TA][\sigma]{l} =
     l\sigma \mathrel{\widehat\PTRSone} \overline{d\sigma} =
     \overline{\interpret[\TA][\sigma]{d}}$, 
    and hence $l \mathrel{\widehat\PTRSone_{\TA}} d$.
    This concludes $\PTRSone \subseteq \widehat\PTRSone_{\TA}$. \qed
  \end{itemize}
\end{proof}
\subsection{\Preb\ Algebras}
As probabilistic $\FS$-algebras are defined so generally, it is not
yet clear how to search them for ones that prove the termination of a
given PTRS.  Now we make one step towards finding probabilistic
algebras, by imposing some conditions to (non-probabilistic)
$\FS$-algebras, so that the relation $\GTD$ can be defined from
orderings which we are more familiar with.

Let us introduce some auxiliary notions.

\begin{definition}[\PrebDom]
  A \emph{\prebdom} is a set $X$ equipped with
  the \emph{\preb\ operation}
  $\E[X] : \FDist{X} \to X$.
\end{definition}

Of particular interest in this work will be the \prebdom s $\Realpos$ and $\Realpos^m$
with \preb\ operations
$\E(\prs{p_1:a_1,\dots;p_n:a_n}) = \sum_{i=1}^n p_i \cdot a_i$.

We naturally generalize the following notions from standard mathematics.
\begin{definition}[Concavity, Affinity]
Let $f : X \to Y$ be a function from and to \prebdom s.
We say $f$ is \emph{concave} with respect to an order ${\GT}$ on $Y$ 
if
$f(\E[X](d)) \GEQ \E[Y](\overline{f(d)})
$
where $\GEQ$ is the reflexive closure of $\GT$.
We say $f$ is \emph{affine} if it satisfies
$f(\E[X](d)) = \E[Y](\overline{f(d)})$.
\end{definition}
Clearly, every affine function is concave.

Now we arrive at the main definition and theorem of this section.

\begin{definition}[\Preb\ $\FS$-Algebra]
  A \emph{\preb\ $\FS$-algebra} is a pair $(\PAone,{\GT})$ of
  an $\FS$-algebra $\PAone$ on a \prebdom\ $X$
  and an order $\GT$ on $X$,
  such that for every $\fs \in \FS$,
  $\I{\fs}$ is monotone and concave with respect to $\GT$.
  We say it is \emph{collapsible} if there exist a concave function
  $\Gcol \colon X \to \Realpos$ (with respect to $>$) and $\epsilon > 0$ such 
  that $\Gcol(\paone) >_\epsilon \Gcol(\patwo)$ whenever $\paone \GT \patwo$.

  We define the relation ${\GT^{\E}} \subseteq X \times \FDist{X}$ by
  $x \GT^{\E} d$ iff $x \GT \E[X](d)$.
\end{definition}

Note that
the following theorem claims soundness but not completeness,
in contrast to Theorem~\ref{t:interpretation}.

\begin{theorem}\label{t:prebary}
A PTRS $\PTRSone$ is \SPAST
if $\PTRSone \subseteq {\GT^{\E}_\PAone}$
for a collapsible \preb\ $\FS$-algebra $(\PAone,{\GT})$.
\end{theorem}
\begin{proof}
  Due to Theorem~\ref{t:interpretation},
  it suffices to show that
  $(\PAone,{\GT^{\E}})$ is a collapsible probabilistic monotone $\FS$-algebra.
  Concerning monotonicity, suppose $\paone \GT^{\E} d$, i.e.,
  $\paone \GT \E[X](d)$,
  and let $\fs \in \FS$. 
  Since $\I\fs$ is monotone and concave with respect to $\GT$ in every argument, we have
  \begin{align*}
    \I{\fs}(\dots,\paone,\dots)
    \GT \I{\fs}(\dots,\E[X](d),\dots)
    \GEQ \E[X](\overline{\I{\fs}(\dots,d,\dots)}) \tpkt
  \end{align*}
  Concerning collapsibility, whenever $x \GT \E[X](d)$
  we have
  \begin{align*}
  \Gcol(x) &>_\epsilon
  \Gcol(\E[X](d)) &&\text{by assumption on $\Gcol$,}\\
  &\ge
  \E(\overline{\Gcol(d)}) &&\text{as $\Gcol : X \to \Real$ is concave with respect to $>$,}\\
  &= \E(\Gcol(d)) &&\text{by the definition of $\E$ on multidistributions.}
  \tag*{\qed}
  \end{align*}
\end{proof}

The rest of the section
recasts two popular interpretation methods, polynomial and matrix interpretations (over the reals), 
as \preb\ $\FS$-algebras.

\paragraph{Polynomial interpretations} were introduced
(on natural numbers \cite{Lankford:75}
and real numbers \cite{Lucas:ITA:05})
for the termination analysis 
of non-probabilistic rewrite systems.
Various techniques for synthesizing polynomial interpretations (e.g., \cite{FGMSTZ:SAT:07}) exist, and these techniques are easily applicable in our setting.

\begin{definition}[Polynomial Interpretation]
A \emph{polynomial interpretation} is an $\FS$-algebra $\PAone$ on $\Realpos$
such that $\fs_\PAone$ is a polynomial for every $\fs \in \FS$.
We say $\PAone$ is \emph{multilinear} if
every $\I{\fs}$ is of the following form
with $\mathsf{c}_{V} \in \Realpos$:
  \[
    \I{\fs}(\seq[n]{\paone}) = \sum_{V \subseteq \sseq[n]{\paone}} \mathsf{c}_{V} \cdot \prod_{\paone_i \in V} \paone_i \tpkt
  \]
\end{definition}

In order to use polynomial interpretations for probabilistic termination,
multilinearity is necessary for satisfying the concavity condition.

\begin{proposition}\label{p:polynomial}
  Let $\PAone$ be a monotone multilinear polynomial interpretation and $\epsilon > 0$.
  If $\Int{l} >_\epsilon \E[](\Int{d})$ for every $l \to d \in \PTRSone$ and $\alpha$,
  then the PTRS $\PTRSone$ is \SPAST.
\end{proposition}
\begin{proof}
  The order $>_\epsilon$ is trivially collapsible with $\Gcol(x) = x$.
  Further,
  every multilinear polynomial is affine and thus concave in all variables.
  Hence $(\PAone,>_\epsilon)$ forms a \preb\ $\FS$-algebra,
  and thus Theorem~\ref{t:prebary} shows that $\PTRSone$ is \SPAST.
  \qed
\end{proof}

  An observation by Lucas~\cite{Lucas:ITA:05} also holds in probabilistic case:
  To prove a finite PTRS $\PTRSone$ \SPAST\ with polynomial interpretations,
  we do not have to find $\epsilon$, but
  it is sufficient to check 
  $l >^{\E}_\PAone d$ for all rules $l \to \distone \in \PTRSone$.
  Define $\epsilon^{}_{l\to d} \defsym \E(\Int{d}) - \Int{l}$ for
  such $\alpha$ that $\alpha(x)=0$.
  Then for any other $\alpha$,
  we can show $\E(\Int{d}) - \Int{l} \ge \epsilon^{}_{l \to d} > 0$.
  As $\PTRSone$ is finite, 
  we can take 
  $\epsilon \defsym \min\{ \epsilon^{}_{l\to d} \mid l \to \distone \in \PTRSone\} > 0$.

\ifextended
\begin{corollary}\label{c:polynomial}
A PTRS $\PTRSone$ is \SPAST if
$\PTRSone \subseteq {>^{\E}_\PAone}$ for a monotone multilinear polynomial interpretation $\PAone$.
\end{corollary}
\fi

\begin{example}[Example~\ref{e:trs} Continued]
  Consider again the PTRS of single rule $\fun{s}(x) \to \prs{p:x; 1-p:\fun{s}(\fun{s}(x))}$.
  Define the polynomial interpretation $\PAone$ by
  $\I{\fun0} \defsym 0$ and $\I{\fun{s}}(x) \defsym x + 1$. 
  Then whenever $p > \frac{1}{2}$ we have
  \begin{multline*}
    \Int{\fun{s}(x)} 
    = x + 1
    > p \cdot x + (1 - p) \cdot (x + 2) 
    = \E(\Int{\prs{p:x; 1-p:\fun{s}(\fun{s}(x))}})
    \tpkt
  \end{multline*}
  Thus, when $p > \frac{1}{2}$ the PTRS is \SPAST by Proposition~\ref{p:polynomial}.
\end{example}

We remark that polynomial interpretations are not covered by 
\cite[Theorem 5]{BG:RTA:05},
since \emph{context decrease}~\cite[Definition 8]{BG:RTA:05} demands
$\Int{\fs(t)} - \Int{\fs(t')} \le \Int{t} - \Int{t'}$,
which disallows interpretations such as $\I\fs(x) = 2x$.


\paragraph{Matrix interpretations}
are
introduced for the termination analysis of term rewriting~\cite{EWZ:JAR:08}.
Now we extend them for probabilistic term rewriting.

\begin{definition}[Matrix Interpretation]
  A \emph{(real) matrix interpretation} is 
  an $\FS$-algebra $\PAone$ on $\Realpos^m$ such that
  for every $f \in \FS$, $\I{\fs}$ is of the form
  \begin{equation}\label{eq:mat}
    \I{\fs}(\seq[n]{\vec x}) = \sum_{i=1}^n C_i \cdot \vec{x}_i + \vec{c} \tkom
  \end{equation}
  where $\vec{c} \in \Realpos^m$,
  and $C_i \in \Realpos^{m \times m}$.
  The order ${\gg_\epsilon} \subseteq \Realpos^m \times \Realpos^m$ is 
  defined by
  \[
    (\seq[m]{\paone})^T \gg_\epsilon (\seq[m]{\patwo})^T \defiff \paone_1 >_\epsilon \patwo_1 \text{ and } \paone_i \geq \patwo_i \text{ for all $i = 2,\dots,m$.}
  \]
\end{definition}

It is easy to derive the following from Theorem~\ref{t:prebary}:
\begin{proposition}\label{p:matrix}
  Let $\PAone$ be a monotone matrix interpretation and $\epsilon > 0$.
  If $\Int{l} \gg_\epsilon \E(\Int{d})$ for every $l \to d \in \PTRSone$ and $\alpha$,
  then the PTRS $\PTRSone$ is \SPAST.
\end{proposition}
\ifextended
\begin{proof}
  The order $\gg_\epsilon$ is collapsible with $\Gcol((\seq[m]x)^T) = x_1$.
  It is well known that \eqref{eq:mat} is affine and thus concave.
  \qed
\end{proof}
\fi

As in polynomial interpretations, for finite systems we do not have to find $\epsilon$.
Monotonicity can be ensured
if \eqref{eq:mat} satisfies $(C_i)_{1,1} \geq 1$ for all $i$, cf.\@ \cite{EWZ:JAR:08}.

\ifextended
Below $\gg$ is defined in the same mannar as $\gg_\epsilon$ but replacing $>_\epsilon$ by $>$.

\begin{corollary}\label{c:matrix}
  Let $\PTRSone$ be a finite PTRS.
  If $\PTRSone \subseteq {\gg^{\E}_\PAone}$ for a monotone matrix interpretation $\PAone$,
  then $\PTRSone$ is \SPAST.
\end{corollary}
\fi

\begin{example}
  \renewcommand{\I}{\interpretation[]}
  \renewcommand{\Int}{\interpret[]}
  Consider the PTRS consisting of the single probabilistic rule 
  \[
    \fun{a}(\fun{a}(x)) \to \prs{p:\fun{a}(\fun{a}(\fun{a}(x))); 1-p:\fun{a}(\fun{b}(\fun{a}(x)))}
    \tpkt
  \]
  Consider the two-dimensional matrix interpretation 
  \begin{align*}
    \I{\fun{a}}(\vec{\paone}) = 
    &
      \begin{bmatrix}
        1 & 1 \\
        0 & 0
      \end{bmatrix}
      \cdot \vec{\paone}
      + 
      \begin{bmatrix} 0 \\ 1 \end{bmatrix}\tkom
    & \I{\fun{b}}(\vec{\paone}) = 
    &
      \begin{bmatrix}
        1 & 0 \\
        0 & 0
      \end{bmatrix}
      \cdot \vec{\paone} \tpkt
  \end{align*}
  Then we have
  \newcommand\bmat[1]{\begin{bmatrix}#1\end{bmatrix}}
  \begin{align*}
    \Int{\fun{a}(\fun{a}(x))}
    = \bmat{x_1+x_2+1\\1}
    & \gg_{1-2p}
    \bmat{x_1+x_2+2p\\1}\\&=
    p \cdot \Int{\fun{a}(\fun{a}(\fun{a}(x)))} + (1 - p) \cdot \Int{\fun{a}(\fun{b}(\fun{a}(x)))}
  \end{align*}
  where $\alpha(x) = (x_1,x_2)^T$.
  Hence this PARS is \SPAST if $p < \frac{1}{2}$, by Proposition~\ref{p:matrix}.
\end{example}
It is worthy of note that the above example cannot be handled with polynomial interpretations, intuitively
because monotonicity enforces the interpretation of the probable reducts $\fun{a}(\fun{a}(\fun{a}(x)))$ and $\fun{a}(\fun{b}(\fun{a}(x)))$ to be greater than that of the left-hand side $\fun{a}(\fun{a}(x))$. 
Generally, polynomial and matrix interpretations are incomparable in strength. 


\ifextended
\section{Implementation}
We extended the termination prover \texttt{NaTT}~\cite{NaTT} with a
syntax for probabilistic rules, and implemented the probabilistic
versions of polynomial and matrix interpretations.
\ifextended

The input format extends the WST format.\footnote{%
\url{https://www.lri.fr/~marche/tpdb/format.html}, accessed November 14, 2017.
}
A probabilistic rewrite rule is specified by
\texttt{$l$ -> $w_1$ : $r_1$ || $\dots$ || $w_n$ : $r_n$},
indicating the probabilistic rewrite rule
\begin{equation}
l \to \prs{\frac{w_1}{w}:r_i,\dots; \frac{w_n}{w}:r_n}
\quad\text{with $w = \sum_{j=1}^n w_j$.}
\label{eq:implement rule}
\end{equation}

The problem of finding interpretations is encoded as
a \emph{satisfiability modulo theory} (SMT) problem and solved by an SMT solver.
We already have an implementation to encode
that $\Int{l} > \Int{r}$ holds for arbitrary $\alpha$, so we only need a little extension to encode
\[
w \cdot \Int{l} > w_1\cdot\Int{r_1}+\dots+w_n\cdot\Int{r_n}
\]
which express the desired orientation condition of probabilistic rule
\eqref{eq:implement rule}.
\else
For usage and implementation details, we refer to the extended version of this paper.
Here we only report that
we tested the implementation on the examples presented in the paper and successfully found
termination proofs.
\fi

The following example would deserve some attention.
\begin{example}
Consider the following encoding of \cite[Figure~1]{FH:POPL:15}:
\def\h{\texttt{?}}
\def\cnt{\texttt{\$}}
\def\s{\texttt{s}}
\def\g{\texttt{g}}
\def\f{\texttt{f}}
\def\z{\texttt0}
\begin{align*}
\h(x) &\to \prs{\half: \h(\s(x)); \half:\cnt(\g(x))} &
\cnt(\z) & \to \prs{1:\z} \\
\h(x) &\to \prs{1:\cnt(\f(x))}&
\cnt(\s(x)) &\to \prs{1:\cnt(x)}
\end{align*}
describing a game where
the player (strategy) can choose either to quit the game and ensure prize $\cnt(\f(x))$,
or to try a coin-toss which on success increments the score
and on failure ends the game with consolation prize $\cnt(\g(x))$.

When $\f$ and $\g$ can be bounded by linear polynomials,
it is possible to automatically prove that the system is \SPAST.
For instance, with
rules for $\f(x) = 2x$ and $\g(x) = \lfloor \tfrac{x}2 \rfloor$, 
\texttt{NaTT} 
(combined with the SMT solver \texttt{z3} version 4.4.1)
found the following polynomial interpretation proving \SPAST:
\begin{align*}
\I\h(x) &= 7x + 11&
\I\s(x) &= x + 1 &
\I\z &= 1 
\\
\I\f(x) &= 3x + 1 &
\I\g(x) &= 2x + 1 &
\I\cnt(x) &= 2x + 1\\
\end{align*}
\end{example}

\section{Conclusion}
This is a study on how much of the classic interpretation-based
techniques well known in term rewriting can be extended to
\emph{probabilistic} term rewriting, and to what extent they remain
automatable. The obtained results are quite encouraging, although
finding ways to combine techniques is crucial if one wants to capture
a reasonably large class of systems, similarly to what happens in
ordinary term rewriting~\cite{Avanzini:Diss:13}. 
Another hopeful future work includes extending our result for proving \AST,
not only \SPAST.


\fi

\bibliographystyle{splncs03}

\end{document}